\documentclass[runningheads]{llncs}
\usepackage[T1]{fontenc}
\usepackage{graphicx}
\usepackage{booktabs}
\usepackage[misc]{ifsym}
\newcommand{\corr}{(\Letter)}
\usepackage{mwe}
\usepackage[american]{babel}

\usepackage{mathtools} 
\usepackage{booktabs} 
\usepackage{tikz} 
\usepackage{algorithm}
\usepackage[table,dvipsnames]{xcolor} 
\usepackage{amsmath,amssymb,amsfonts}
\usepackage{algpseudocode}
\usepackage{textcomp}
\usepackage{comment}
\usepackage{textcomp}
\usepackage{booktabs}
\usepackage{lipsum}
\usepackage{wrapfig}
\usepackage{pgfplots}
\usepackage{tikz}
\usepackage{scalerel}
\usepackage{enumitem}
\usetikzlibrary{positioning}

\newcommand{\RE}{\mathbb{R}}
\newcommand{\ZZ}{\mathbb{Z}}
\newcommand{\best}[1]{{\cellcolor{green!15}{#1}}}

\tikzset{
red node/.style={circle,fill=red!20,draw,minimum size=0.5cm,inner sep=0pt},
green node/.style={circle,fill=green!20,draw,minimum size=0.5cm,inner sep=0pt},
blue node/.style={circle,fill=blue!20,draw,minimum size=0.5cm,inner sep=0pt},
}

\algblock{DOPAR}{ENDDOPAR}
\algnewcommand\algorithmicdopar{\textbf{parallel for}}
\algnewcommand\algorithmicdopardo{\textbf{do}}
\algnewcommand\algorithmicenddopar{\textbf{end}}
\algrenewtext{DOPAR}[1]{\algorithmicdopar\ #1\ \algorithmicdopardo}
\algrenewtext{ENDDOPAR}{\algorithmicenddopar}

\newcommand{\calO}{\mathcal{O}}

\begin{document}

\title{Scaling Weisfeiler–Leman Expressiveness Analysis \\ to Massive Graphs with GPUs}
\author{Filippo Biondi\inst{1} \corr \and Mirco Tribastone\inst{1} \and Max Tschaikowski\inst{2}}

\toctitle{Scaling Weisfeiler–Leman Expressiveness Analysis to Massive Graphs with GPUs}
\tocauthor{Filippo Biondi, Mirco Tribastone, Max Tschaikowski}

\institute{
IMT Lucca, Italy\\
\email{\{filippo.biondi,mirco.tribastone\}@imtlucca.it}
\and
Sapienza University Rome, Italy\\
\email{max.tschaikowski@uniroma1.it}
}

\titlerunning{Scaling 1-WL Analysis to Massive Graphs with GPUs}

\maketitle              

\begin{abstract}
The stable coloring of the Weisfeiler-Leman (1-WL) test  is a cornerstone of Graph Neural Networks because it provides an upper bound to the expressive power of message-passing architectures. Unfortunately, computing it presents two fundamental bottlenecks. First, classic algorithms are inherently sequential and cannot exploit modern massively parallel hardware. Second, these are \emph{global} algorithms, i.e., they require availability in memory of the full graph, severely limiting applicability to real-world instances. We leverage a linear-algebraic interpretation of 1-WL stable coloring and introduce two key contributions: (i)~a randomized refinement algorithm with tight probabilistic guarantees and (ii)~a correctness-preserving batching scheme that decomposes the graph into independently processable subgraphs while provably returning a stable coloring of the original graph. This approach maps directly to GPU-efficient primitives. In numerical experiments, our CUDA implementation delivers speedups up to two orders of magnitude over classical CPU-based partition refinement and, for the first time, successfully computes stable colorings on web-scale graphs with over 30 billion edges, where CPU baselines time out or fail.
\end{abstract}

 \section{Introduction}

Graph Neural Networks (GNNs) have emerged as powerful tools for learning on graph-structured data, achieving state-of-the-art performance in domains such as chemistry, biology, social networks, and recommender systems~\cite{10.1145/2487788.2488173,KAUFFMAN1969437,10.1145/2049662.2049663,Shizuka22}. A central question in their theoretical analysis is \emph{expressiveness}: which structural distinctions can a GNN architecture capture? The Weisfeiler-Leman (WL) graph isomorphism test---particularly its first-order variant (1-WL)---has become the standard lens for characterizing GNN expressiveness~\cite{DBLP:conf/lics/Grohe21,JMLR:v24:22-0240}. 
The 1-WL \emph{stable coloring} is relevant because it provides an upper bound to the expressive power of  message-passing GNNs---two nodes assigned the same color will always receive identical embeddings, meaning the network cannot distinguish  them.

Computing a stable coloring amounts to determining a partition of a graph  such that nodes in the same block share identical neighborhood counts~\cite{godsil1997compact,DBLP:conf/lics/Grohe21}.  Classical algorithms based on partition refinement compute the coarsest stable coloring, i.e., the partition with the fewest blocks, in optimal $O(m \log n)$ complexity~\cite{paige1987three}, a bound that is asymptotically tight~\cite{Grohe2013}. Yet they suffer from two  fundamental limitations: they are \emph{sequential} and \emph{global}, requiring the entire input graph to reside in memory.  These characteristics make them ill-suited for modern architectures such as GPUs, which excel at massively parallel workloads.

As a consequence, despite its centrality, computing stable colorings at web scale remains practically infeasible, limiting empirical WL-based expressiveness analysis to small or medium graphs. In this paper we address both limitations of the state of the art with a \emph{parallel} algorithm that can execute \emph{locally} on independent subgraphs.

\textbf{Randomized linear-algebraic refinement.}
Our approach leverages a key invariance property: if a vector is symmetric on an equitable partition (that is, it has equal values on its blocks of nodes), then so is its image under the adjacency matrix~\cite{10.5555/2892753.2892817,DBLP:conf/cdc/TognazziTTV18,DBLP:journals/tcs/CardelliTTV19a,DBLP:conf/icdm/SquillaceTTV24}. As a result, stable coloring can be expressed as a sequence of matrix–vector multiplications (computing images until a fixed point) that naturally maps to highly optimized GPU kernels. 

While linear-algebraic approaches based on deterministic power iteration have been explored previously~\cite{10.5555/2892753.2892817}, they typically rely on floating-point arithmetic and exact constructions that may suffer from numerical instability and limited scalability on very large graphs. In contrast, we develop a randomized linear-algebraic refinement Monte Carlo framework that operates entirely in integer arithmetic with explicit probabilistic correctness guarantees. 

\textbf{Correctness-preserving batching for large-scale graphs.}
The global nature of 1-WL means that color refinement depends on complete neighborhood counts. Processing arbitrary edge subsets independently therefore breaks correctness in general.
We introduce a batching scheme that provably preserves the stable coloring of the original graph. We partition edges into disjoint batches and identify inner nodes, whose outgoing edges lie entirely within a single batch. Only such nodes are eligible for color aggregation within that batch; nodes with cross-batch edges are temporarily kept as singleton colors.
Each batch thus induces a local refinement that is guaranteed to be consistent with the global stable coloring. The resulting local colorings are merged via a quotient construction, yielding a stable coloring of the original graph—though not necessarily the coarsest one. Iterating this procedure produces a fixed point: if the reduced graph at the fixed point fits in a single batch, the resulting coloring is provably the coarsest stable coloring; otherwise, it is a refinement thereof. This enables correct 1-WL analysis even when the full graph does not fit in device memory.

\textbf{Experimental results.} Our two main contributions are orthogonal: the randomized refinement addresses compute scalability, while batching addresses memory scalability. We present extensive experiments where we analyze both aspects using a CUDA implementation. On medium-scale and web-scale graphs, randomized refinement achieves up to $\sim$138$\times$ speedups over classic CPU-based partition-refinement algorithms. By simulating memory constraints on graphs that otherwise could be fully analyzed, iterated contraction through a fixed number of batches yields partitions that are empirically within 5\% of the coarsest stable coloring. On graphs with more than 30 billion edges, for which no coarsest stable coloring is available, combining batching and randomization enabled 1-WL stable colorings at previously infeasible scales.

\section{Related Work}

\textbf{Weisfeiler--Leman and graph expressiveness.} 
The Weisfeiler–Leman (WL) algorithm, in particular its first-order variant (1-WL), is central in the theoretical understanding of GNN expressiveness~\cite{xu2018how,10.1609/aaai.v33i01.33014602,DBLP:conf/lics/Grohe21}. The 1-WL stable coloring characterizes the  power of message-passing architectures: nodes receiving the same stable color cannot be distinguished by such networks. A linear-algebraic characterization of stable coloring was established by~\cite{10.5555/2892753.2892817}, where color refinement is shown to be expressible via invariance under adjacency matrix multiplication. However, it relies on floating-point power iteration and assumes that the full graph fits in memory. It does not address memory scalability or batching.

\textbf{Parallel partition refinement algorithms.} 
Stable coloring coincides with the computation of the coarsest equitable partition of a graph~\cite{godsil1997compact,DBLP:conf/lics/Grohe21}. When a graph is interpreted as a dynamical system via its adjacency matrix, equitable partitions correspond to backward equivalence (BE) for linear systems~\cite{DBLP:journals/tcs/CardelliTTV19a,DBLP:journals/tac/CardelliGLTTV23}, an observation that has been also extended to nonlinear systems~\cite{10.1145/3071178.3071265,DBLP:conf/qest/CardelliTTV18,DBLP:journals/jlap/CardelliSTTV23,BBLTTV21Lics}. 
Furthermore, because of the relation between equitable partitions and probabilistic bisimulation for Markov chains~\cite{Larsen19911}, computing equitable partitions  is known to be P-complete~\cite{DBLP:conf/fossacs/ChenBW12}, suggesting that efficient polylogarithmic-time parallel algorithms are unlikely in the worst case. Distributed and multicore CPU implementations exist~\cite{DBLP:journals/entcs/BlomHKP08,DBLP:conf/tacas/DijkP16}, but they rely on irregular memory access patterns that are poorly suited to GPU architectures. 

\textbf{Randomized partition refinement algorithms.} 
Monte Carlo variants of partition refinement based on the Schwartz–Zippel lemma~\cite{schwartz1980fast} have been proposed in more general algebraic settings~\cite{DBLP:conf/lics/ArgyrisLLTTV23}. These approaches operate over fields and rely on standard polynomial identity testing results. In contrast, we work over the ring  $\ZZ / 2^{64} \ZZ$, which enables efficient integer arithmetic on GPUs but does not satisfy the assumptions of classical Schwartz–Zippel bounds. We therefore derive a tailored probabilistic guarantee adapted to this ring structure.
 
\textbf{Power iteration refinement algorithms.} 
Kersting et al. proposed a linear-algebraic formulation of color refinement based on deterministic power iteration~\cite{10.5555/2892753.2892817}. Their method constructs specific real-valued vectors (e.g., using logarithms of prime numbers) whose images under repeated adjacency multiplication induce the coarsest stable coloring. While elegant, this approach relies on floating-point arithmetic and exact numerical comparisons, which may suffer from precision loss on large graphs. As we demonstrate in Section \ref{sec:numerical}, this can lead to incorrect color merges at scale.
In contrast, our method employs a Monte Carlo refinement scheme over fixed-width modular integer arithmetic, providing explicit probabilistic correctness guarantees. Moreover, our formulation is designed to integrate naturally with batching for memory scalability, which is not addressed in prior power-iteration approaches.

\textbf{GPU graph processing frameworks.} 
General-purpose GPU graph frameworks such as Gunrock~\cite{DBLP:journals/topc/WangPDWYWOYLRO17}, GraphBLAST~\cite{DBLP:journals/toms/YangBO22}, and NVIDIA’s cuGraph provide optimized primitives for traversal and classical vertex coloring. However, they do not implement WL stable coloring and typically assume that the entire graph fits in GPU memory. Our work complements these systems with a stable coloring algorithm that supports correctness-preserving batching for graphs exceeding device memory.

\section{Background}

We consider a directed graph $G = (V, A)$ with node set $V = {1,\dots,n}$ and adjacency matrix $A \in \mathbb{N}_0^{n \times n}$. The entry $a_{i,j}$ denotes the (possibly weighted) number of edges from node $i$ to node $j$. This formulation naturally covers both unweighted graphs ($a_{i,j} \in \{0,1\}$) and weighted graphs. Even if the original graph is unweighted, weighted adjacency matrices arise after aggregation in the quotient construction defined below.

A \emph{partition} $\mathcal{H} = {B_1, \dots, B_k}$ of $V$ is a collection of pairwise disjoint nonempty subsets (called \emph{blocks}) such that $\bigcup_{i=1}^k B_i = V$.

\begin{definition}[Stable Coloring]
A partition $\mathcal{H}$ of $V$ is a \emph{stable coloring} if for every pair of blocks $B, B' \in \mathcal{H}$ and all nodes $i,j \in B$,
\[
\sum_{l \in B'} a_{i,l}
=
\sum_{l \in B'} a_{j,l}.
\]
\end{definition}

In the unweighted case, this condition states that nodes $i$ and $j$ have the same number of neighbors in each block $B'$. For weighted graphs, it requires equality of total outgoing weight toward every block.

Stable colorings coincide with equitable partitions in algebraic graph theory and with \emph{backward equivalence (BE)} in the dynamical-systems literature~\cite{DBLP:journals/tcs/CardelliTTV19a}. Throughout the paper, we use the term \emph{stable coloring} as the primary notion, but we rely on these equivalent characterizations when invoking existing invariance results.

\begin{definition}[Coarseness]
Given two partitions $\mathcal{H}$ and $\mathcal{H}'$ of $V$, we say that $\mathcal{H}$ is \emph{coarser} than $\mathcal{H}'$ if every block of $\mathcal{H}'$ is contained in a block of $\mathcal{H}$.
\end{definition}

\begin{definition}[Coarsest Stable Coloring]
A stable coloring $\mathcal{H}$ is the \emph{coarsest stable coloring} if it is coarser than every other stable coloring of $V$.
\end{definition}

\begin{definition}[Quotient Graph]\label{def:qg}
Given a stable coloring $\mathcal{H} = {B_1, \dots, B_k}$, we construct the \emph{quotient graph} $G' = (\mathcal{H}, A')$, whose nodes correspond to the blocks of $\mathcal{H}$. The aggregated adjacency matrix $A' \in \mathbb{N}_0^{k \times k}$ is defined by
\begin{equation}\label{eq:sc}
a'_{B,B'} = \sum_{l \in B'} a_{i,l},
\end{equation}
for any fixed representative $i \in B$.
\end{definition}
The definition is well-posed because the stability condition ensures that the right-hand side does not depend on the choice of representative $i \in B$. Even if the original adjacency matrix $A$ is binary, the quotient matrix $A'$ is generally weighted: the entry $a'_{B,B'}$ counts the number of edges from nodes in block $B$ to nodes in block $B'$.

The quotient graph preserves stable-coloring structure in the following sense: if $\mathcal{H}$ is a stable coloring of $G$, then further refinement can equivalently be performed on the reduced graph $G'$. This property underlies both classical partition-refinement algorithms and the batching strategy introduced later.

\subsubsection{Linear-Algebraic Characterization.}
Stable colorings admit a convenient linear-algebraic interpretation. Let
\begin{align*}
U_{\mathcal{H}} = \{ x \in \mathbb{R}^n \mid 
x_i = x_j \text{ whenever } i,j \in B \text{ for some } B \in \mathcal{H} \}
\end{align*}
denote the subspace of vectors that are constant on each block of $\mathcal{H}$. A partition $\mathcal{H}$ is a stable coloring if and only if
\begin{equation}\label{eq:inv}
A(U_{\mathcal{H}}) \subseteq U_{\mathcal{H}},
\end{equation}
that is, the image of any block-constant vector under multiplication by $A$ remains block-constant~
\cite{DBLP:journals/tcs/CardelliTTV19a}.

This perspective provides the foundation for the linear-algebraic refinement algorithm of Section~4 and enables an implementation based on matrix-vector multiplication primitives that are well suited to GPU architectures.

\section{Methods}
\label{sec:methods}

We present two complementary algorithms for scaling 1-WL stable coloring.
The first (Section~\ref{sec:randomized}) replaces the classical sequential partition-refinement
loop with a randomized linear-algebraic procedure whose inner loop reduces to a single
matrix–vector multiplication—a primitive that maps directly to GPU hardware.
The second (Section~\ref{sec:batching}) introduces a correctness-preserving decomposition of
the edge set into independently processable batches, enabling the computation of a
stable coloring even when the graph does not fit in device memory.
The two contributions are \emph{orthogonal}: the randomized scheme addresses
\emph{compute} scalability (speed), while batching addresses \emph{memory} scalability
(capacity).

\subsection{Randomized Stable Coloring}
\label{sec:randomized}

\textbf{Overview.}
The linear-algebraic characterization of Eq.~\ref{eq:inv} suggests an iterative refinement strategy: 
given a candidate partition $H$, draw a random vector $w$ that is symmetric on $H$, multiply by $A$, and check whether
$Aw$ is still symmetric on $H$.
If it is not, $Aw$ witnesses that $H$ must be refined—specifically, the coarsest
partition on which \emph{both} $w$ and $Aw$ are simultaneously symmetric is a
strict refinement of $H$ that is still consistent with the true stable coloring.
Iterating this process until no further refinement is possible yields a stable coloring. This is the key idea behind Algorithm~1.

The attractive feature of this approach is that each iteration reduces to (i)~sampling
a random symmetric vector and (ii)~performing a single sparse matrix–vector
multiplication, followed by a lexicographic sort to read off the new partition.
Both operations are efficient on GPU hardware.

\textbf{Initial partition.} Algorithm 1 accepts an arbitrary initial partition $H$
and computes the coarsest stable coloring that refines it. For computing the coarsest stable coloring of the full graph, $H$ is set to the trivial partition where all nodes are in a single block. The ability to start from a nontrivial initial partition is essential for the batching scheme: there, each batch must enforce that only inner nodes — those whose outgoing edges lie entirely within the batch — are eligible for aggregation. This is achieved by initializing with a partition that places every boundary node in its own singleton block, so that refinement within a batch can never merge nodes whose neighborhood information is incomplete.

\textbf{The challenge of floating-point arithmetic.}
Over the reals $\mathbb{R}$, Algorithm~1 is  justified by the
invariance characterization of Eq.~\ref{eq:inv}: if $H$ is not stable, there
exists a real-valued vector $w$ symmetric on $H$ such that $Aw$ is \emph{not}
symmetric on $H$.
A concrete instantiation using logarithms of primes was proposed by \cite{10.5555/2892753.2892817}: their vectors are constructed so that $Aw$ always refines $H$
when $H$ is not stable, and the process terminates at the coarsest stable coloring.

However, floating-point arithmetic introduces numerical errors, and exact comparisons of real-valued products may become unreliable at scale.
Indeed, Section~\ref{sec:reduction} shows that the method of \cite{10.5555/2892753.2892817} fails to
compute the correct coarsest stable coloring on several of our benchmark graphs,
silently merging blocks that should remain distinct.

The natural remedy presented in Algorithm~1 is to work in \emph{integer arithmetic}, which is exact. We therefore perform all computations in the modular ring
$R = \mathbb{Z}/2^{64}\mathbb{Z}$, using unsigned 64-bit integers with
overflow as the arithmetic domain.

\textbf{The ring issue.}
Over the ring $R=\mathbb{Z}/2^{64}\mathbb{Z}$, the invariance
characterization of Eq.~\eqref{eq:inv} can fail in degenerate
cases: nonzero polynomials may have many roots, so a random probe
vector could miss a required refinement. However, our benchmark graphs
satisfy $n\leq 10^9\leq 2^{64/2}$, bounding all row sums of~$A$
by~$2^{64/2}$. This structural constraint allows us to derive a tight
correctness guarantee over the ring via the first isomorphism theorem, 
bypassing the classical Schwartz--Zippel lemma that
requires a field or its generalization (see~\cite{Bishnoi2018OnZeros}), whose assumptions are not satisfied by our ring $R = \ZZ / 2^{64} \ZZ$. 

The next result ensures that a randomly chosen symmetric vector
detects any missing refinement with failure probability at most
$10^{-9}$, provided that row sums of $A$ are bounded by $2^{q/2}$ over the ring $R = \ZZ / 2^{q} \ZZ$ and $q \geq 64$. The bound $2^{q/2} = \sqrt{2^q}$ is informed by the generalized birthday paradox and cannot be relaxed substantially. Proofs and implementation details (pseudorandom generation, sorting, sparse multiplication) are deferred to the paper appendix. 

\begin{algorithm}[t]
\caption{Randomized Stable Coloring (RSC)}
\label{alg:rsc}
\begin{algorithmic}[1]
\Require Adjacency matrix $A \in \mathbb{N}_0^{n\times n}$;
         initial partition $H$ of $V$;
         precision parameter $q \ge 1$ defining
         $R = \mathbb{Z}/2^q\mathbb{Z}$
\Ensure  A stable coloring $H^*$ refining $H$, with error probability $\le 10^{-9}$
\State $H' \leftarrow \emptyset$
\While{$H \ne H'$}
    \Repeat
        \State Sample $w \in R^n$ uniformly at random among vectors symmetric on $H$
    \Until{$w$ attains exactly $|H|$ distinct values}   
    \State Compute $z = Aw \bmod 2^q$
    \State $H' \leftarrow$ coarsest partition refining $H$ on which $z$ is symmetric
    \State $H \leftarrow H'$
\EndWhile
\State \Return $H$
\end{algorithmic}
\end{algorithm}

\begin{theorem}[Randomized Stable Coloring]
\label{thm:rsc}
Let $A \in \mathbb{N}_0^{n\times n}$ with $n \le 2^{q/2}$ and with all row sums
bounded by $2^{q/2}$, where we work over $R = \mathbb{Z}/2^{q}\mathbb{Z}$ with $q \geq 64$.
Then Algorithm~\ref{alg:rsc} computes the coarsest stable coloring refining $H$ with error
probability at most $10^{-9}$.
Its expected worst-case time complexity is $O(nm)$, where $m$ is the number of
nonzero entries of $A$, and its space complexity is $O(m)$.
\end{theorem}

We note that the expected worst-case complexity $O(mn)$ is inferior to the optimal
$O(m\log n)$ of classical splitter-based partition refinement.
However, the memory access pattern of Algorithm~\ref{alg:rsc} are
substantially more GPU-friendly, which in practice yields the large speedups (see Section~\ref{sec:experiments}).

\subsection{Batched Stable Coloring}\label{sec:batching}

Computing a stable coloring requires \emph{global} neighborhood
information: deciding whether two nodes belong to the same block
demands comparing their aggregate edge counts toward every block of the
current partition. When the edge set of the graph exceeds device
memory, this global dependency appears to preclude any decomposition
into independently processable subproblems.  We show that a careful
partitioning of the edge set into \emph{batches}, combined with a
distinction between \emph{inner} and \emph{boundary} nodes, yields
local refinements that are provably consistent with the global stable
coloring.

\subsubsection{Inner and boundary nodes.}
Let $G=(V,A)$ be a graph with edge set
$E = \{(i,j) \mid a_{i,j}\neq 0\}$.  Given a partition of $E$ into
$N_b$ disjoint batches $\mathcal{E} = \{E_1,\ldots,E_{N_b}\}$, we
classify nodes according to how their outgoing edges are distributed
across batches.

\begin{definition}[Inner and boundary nodes]
\label{def:inner-boundary}
For an edge partition $\mathcal{E}=\{E_1,\ldots,E_{N_b}\}$ of~$E$,
the \emph{inner nodes} of batch~$E_k$ are
\[
  V_k = \bigl\{\, i\in V \;\bigm|\; \nexists\,(i,j)\in E\setminus E_k \,\bigr\},
\]
i.e., every outgoing edge of $i$ belongs to $E_k$.%
\footnote{Without loss of generality, we assume every node has at
least one outgoing edge; if not, replace $A$ by $A+I$, which
preserves the set of stable colorings.}
The \emph{boundary nodes} are $\mathcal{B} = V \setminus
\bigcup_{k=1}^{N_b} V_k$.
\end{definition}

A key structural property is that the inner-node sets
$V_1,\ldots,V_{N_b}$ are \emph{pairwise disjoint}: if node $i$ has
all outgoing edges in~$E_k$, it cannot have outgoing edges in any
other batch, so $i\notin V_{k'}$ for $k'\neq k$.

\subsubsection{Local refinement with controlled initialization.}
For each batch $E_k$, consider the restricted graph $(V, A|_{E_k})$,
where $A|_{E_k}$ retains only those entries of~$A$ corresponding to
edges in~$E_k$. Applying Algorithm~\ref{alg:rsc} to this restricted
graph can only correctly evaluate the stability condition
(Eq.~\ref{eq:sc}) for inner nodes of batch~$k$,
since only their complete outgoing neighborhoods are available.

This is precisely where the \emph{initial partition} parameter of
Algorithm~\ref{alg:rsc} becomes essential. We initialize each batch
with the partition
\begin{equation}\label{eq:batch-init}
  \mathcal{H}^0_k \;=\; \bigl\{\,V_k\,\bigr\}
    \;\cup\; \bigl\{\,\{i\} \;\bigm|\; i\in V\setminus V_k\,\bigr\},
\end{equation}
which places all inner nodes in a single block and assigns every
non-inner node (including all boundary nodes and inner nodes of other
batches) to its own singleton block. Since Algorithm~\ref{alg:rsc}
can only \emph{refine} its initial partition—splitting blocks but
never merging them—this initialization enforces two invariants:
\begin{enumerate}[label=(\roman*),nosep]
  \item only nodes whose \emph{complete} outgoing neighborhood is
    present in the batch can be merged into the same block;
  \item the refinement within each batch is monotone with respect to
    the global stable coloring.
\end{enumerate}

\subsubsection{Merging local refinements.}
Let $\mathcal{H}_k$ denote the stable coloring returned by
Algorithm~\ref{alg:rsc} on $(V,A|_{E_k})$ with initialization
$\mathcal{H}^0_k$.  We construct a global partition~$\mathcal{H}$ by
collecting, from each batch~$k$, the blocks of~$\mathcal{H}_k$ that
are subsets of~$V_k$, and placing every boundary node in its own
singleton block:
\begin{equation}\label{eq:merge}
  \mathcal{H} \;=\;
    \bigcup_{k=1}^{N_b}\,\bigl\{B\in\mathcal{H}_k \mid B\subseteq V_k\bigr\}
    \;\cup\;
    \bigl\{\{i\} \mid i\in\mathcal{B}\bigr\}.
\end{equation}

\begin{theorem}[Correctness of batched refinement]
\label{thm:batch-correct}
Under the assumptions of Theorem~\ref{thm:rsc}, let $\mathcal{E}$ be
an arbitrary edge partition of~$E$, and let $\mathcal{H}$ be the
partition constructed in~\eqref{eq:merge}. Then $\mathcal{H}$ is a
stable coloring of~$(V,A)$, with error probability at most~$10^{-9}$.
\end{theorem}

\subsubsection{Iterated contraction via quotient graphs.}
A single round of batched refinement yields a valid stable coloring
$\mathcal{H}$ that is, in general, \emph{finer} than the coarsest
stable coloring of~$G$: boundary nodes are conservatively isolated as
singletons, potentially preventing aggregations that the global
algorithm would discover. 

To recover a coarser partition, we iterate.
After each round, we construct the quotient graph $G' =
(\mathcal{H}, A')$ as per Definition~\ref{def:qg}. It
preserves stable-coloring structure: any further refinement of the
original graph can equivalently be performed on~$G'$. Crucially, the
quotient is typically much smaller than the original graph, so
subsequent rounds operate on progressively smaller instances.

Algorithm~\ref{alg:rbsc} iterates this contraction loop until one of
two termination conditions is met:
\begin{enumerate}[label=(\roman*),nosep]
  \item $N_b = 1$: the (reduced) graph fits in a single batch, so
    Algorithm~\ref{alg:rsc} runs on the full graph and returns the
    coarsest stable coloring with high probability; or
  \item no further reduction occurs ($|V|$ is unchanged), indicating
    a fixed point of the batched contraction.
\end{enumerate}
In case~(i), the result is provably the coarsest stable coloring.  In
case~(ii), the result is a valid stable coloring that may be strictly
finer. Our experiments (Section~\ref{sec:experiments}) show that, in
practice, the gap is small: partitions by multiple batches
are consistently within 5\% of the coarsest stable coloring.

The merged partition~$\mathcal{H}$ in~\eqref{eq:merge} is
independent of batch processing order, since each batch induces a
monotone refinement and the merge computes their meet on the partition
lattice. The for loop at line~7 of Algorithm~\ref{alg:rbsc} is
embarrassingly parallel: batches operate on disjoint edge subsets and
are scheduled concurrently across GPUs in our implementation.

\begin{algorithm}[t]
\caption{Randomized Batched Stable Coloring (RBSC)}
\label{alg:rbsc} 
\begin{algorithmic}[1]
\Require Adjacency matrix $A\in\mathbb{N}_0^{n\times n}$;
  precision $q\geq 1$ defining $R=\mathbb{Z}/2^q\mathbb{Z}$;
  maximum batch size $M_b\in\mathbb{N}$
\Ensure A stable coloring of $(V,A)$
\Repeat
  \State $\mathcal{H}\gets\varnothing$
  \State $E\gets\{(i,j)\mid a_{i,j}\neq 0\}$;\quad
         $n\gets|V|$;\quad $m\gets|E|$
  \State $N_b\gets\lceil m/M_b\rceil$
  \State Partition $E$ into $E_1,\ldots,E_{N_b}$
         with $|E_k|\leq M_b$
  \State Compute $V_1,\ldots,V_{N_b}$ per
         Definition~\ref{def:inner-boundary}
  \For{$k=1,\ldots,N_b$} 
    \State $\mathcal{H}^0_k \gets \{V_k\}\cup\{\{i\}\mid
           i\in V\setminus V_k\}$
    \State $\mathcal{H}_k\gets
           \textsc{RSC}(V,\,A|_{E_k},\,\mathcal{H}^0_k)$
    \State $\mathcal{H}\gets\mathcal{H}\cup
           \{B\in\mathcal{H}_k\mid B\subseteq V_k\}$
  \EndFor
  \State $\mathcal{B}\gets V\setminus\bigcup_{B\in\mathcal{H}} B$
  \State $\mathcal{H}\gets\mathcal{H}\cup
         \{\{i\}\mid i\in\mathcal{B}\}$
  \State Construct quotient $(V,A)\gets
         (\mathcal{H},A')$ (Def.~\ref{def:qg})
\Until{$N_b=1$ \textbf{or} $|V|=n$}
\State \Return $\mathcal{H}$
\end{algorithmic}
\end{algorithm}

\begin{theorem}[Complexity of RBSC]\label{thm:rbsc-complexity}
Under the assumptions of Theorem~\ref{thm:rsc}, the expected
worst-case time complexity of Algorithm~\ref{alg:rbsc} is $O(n^2 m)$,
and its space complexity is $O(m)$.
\end{theorem}

\section{Numerical Experiments}\label{sec:experiments}

Our two algorithmic contributions---randomized linear-algebraic
refinement (Section 4.1) and correctness-preserving batching
(Section 4.2)---address orthogonal bottlenecks: the
former targets compute scalability (speed), the latter memory
scalability (capacity).  We evaluate them along three corresponding
axes:
\begin{enumerate}[label=(\roman*),nosep]
  \item \emph{Correctness and reduction quality}: does the randomized
    algorithm recover the coarsest stable coloring, and how close does
    batching come to it?
  \item \emph{Runtime efficiency}: how do GPU-accelerated
    implementations compare to state-of-the-art CPU baselines?
  \item \emph{Scalability to massive graphs}: can we compute stable
    colorings on instances whose edge sets far exceed the memory of a
    single device?
\end{enumerate}

\subsection{Experimental Set-up}
\label{sec:setup}

\subsubsection{Datasets.}
We use eleven medium- to large-scale graphs from established
collections: Flickr, Yelp, and Reddit from the PyTorch Geometric
dataset collection~\cite{DBLP:journals/corr/abs-1903-02428}, and eight web graphs from the
WebGraph framework~\cite{BoVWFI,BCSU3,BRSLLP,BMSB}.  These range from fewer than
$10^5$ nodes to over $4\times 10^7$ nodes and $1.5\times 10^9$ edges
(Table~\ref{tab:tablereduction}).  For the scalability experiments of
Section~\ref{sec:massive}, we additionally consider three web-scale
graphs (\texttt{gsh-2015}, \texttt{clueweb12}, \texttt{uk-2014}) with
$33$--$48$ billion edges, for which no baseline is able to compute a
stable coloring.

\subsubsection{Implementation.}
\label{sec:implementation}
Our algorithms are implemented in C++ with CUDA for GPU
kernels and OpenMP for scheduling batches across devices.
Sparse matrix-vector multiplications use device-wide reduction
primitives from the CUB library.  Random vectors are generated with
64-bit pseudorandom number generators (xoshiro256**
with SplitMix64 seeding~\cite{10.1145/3460772,steele14}
), ensuring a per-iteration error probability below
$10^{-9}$ (Theorem~\ref{thm:rsc}).

\subsubsection{Baselines.}
We compare against three baselines. (i) \textbf{PR}: the partition-refinement algorithm of~\cite{DBLP:conf/tacas/ValmariF10}, 
    executed as a
    single-threaded CPU program.  This implements the classical
    $O(m\log n)$ splitter-based approach. (ii) \textbf{RSC-CPU}: our own randomized
    Algorithm~\ref{alg:rsc}, executed on a single CPU core without
    batching, to isolate the effect of GPU acceleration from the
    algorithmic reformulation. (iii) \textbf{SR$_c$}: the multicore symbolic reduction
    of~\cite{DBLP:journals/entcs/BlomHKP08,DBLP:conf/tacas/DijkP16}, using $c\in\{32,64\}$
    CPU cores.
\subsubsection{Hardware.}
All experiments were conducted on a server equipped with a 192-core
AMD EPYC~7R13 CPU, 2\,TiB of RAM, and 8 NVIDIA H200 GPUs
(141\,GB VRAM each), running Ubuntu~24.04.

\subsubsection{Metrics.}
We report three quantities: (i)~\emph{reduction power}, defined as
the ratio $|\mathcal{H}|/|V|$ between the number of blocks in the
computed stable coloring and the number of nodes in the input graph
(lower is better); (ii)~wall-clock runtime excluding I/O; and
(iii)~success or failure within a fixed timeout (3\,600\,s for medium-scale
graphs, 10\,800\,s for web-scale graphs).

\subsubsection{Batch configurations.}
We evaluate Algorithm~\ref{alg:rbsc} with $N_b\in\{1,2,4\}$ batches.
The single-batch setting ($N_b=1$) processes the full edge set at
once, so RBSC reduces to the batch-free Algorithm~\ref{alg:rsc} and
returns the coarsest stable coloring with high probability.  Settings
$N_b=2$ and $N_b=4$ simulate scenarios where only a fraction of the
edge set fits in device memory.

\subsubsection{Device usage.}
When $N_b=1$, a single batch is processed on a single GPU.  For
$N_b>1$, the resulting batches are scheduled
concurrently across available GPUs (one batch per device), with merge
and quotient construction performed on the host between outer
iterations. When $N_b=1$, execution uses a single GPU. For $N_b>1$, batches run concurrently across devices (one per GPU), with host-side merge between iterations. Processing the full graph in one batch is fastest when VRAM permits; multi-GPU scheduling compensates when it does not.

\subsection{Correctness and Reduction Quality}
\label{sec:reduction}

Before evaluating performance, we verify that the randomized algorithm
is empirically correct and that batching preserves partition quality.
\begin{table*}[t!]
    \centering
    \begin{tabular}{rll @{\hspace{1em}}c @{\hspace{2em}}c @{\hspace{1em}}c} \toprule
    &&& \multicolumn{3}{c}{\textbf{Randomized Batched SC (RBSC)}} \\
    \cmidrule(l){4-6}
    \multicolumn{1}{c}{\emph{No.}} & \emph{Dataset Name} & \emph{Size} & \multicolumn{1}{c}{$N_b = 1$} & \multicolumn{1}{c}{$N_b = 2$} & \multicolumn{1}{c}{$N_b = 4$} \\ 
    &&& (coarsest) \\ \midrule
    1 &    {  flickr}  & \begin{tabular}{@{}l@{}}   $n$: 89,250 \\   $m$: 899,756 \end{tabular} & 90.91\% & 91.26\% & 91.80\% \\ \midrule
    2 &    {  cnr-2000}  & \begin{tabular}{@{}l@{}}   $n$: 325,557 \\   $m$: 3,216,152 \end{tabular} & 26.24\% & 26.24\% & 26.24\% \\ \midrule
    3 &    {  yelp}  & \begin{tabular}{@{}l@{}}   $n$: 716,847 \\   $m$: 13,954,819 \end{tabular} & 92.47\% & 94.73\% & 96.24\% \\ \midrule
    4 &    {  ljournal-2008}  & \begin{tabular}{@{}l@{}}   $n$: 5,363,260 \\   $m$: 79,023,142 \end{tabular} & 80.09\% & 80.38\% & 80.56\% \\ \midrule
    5 &    {  hollywood-2009}  & \begin{tabular}{@{}l@{}}   $n$: 1,139,905 \\   $m$: 113,891,327 \end{tabular} & 53.28\% & 53.28\% & 53.28\% \\ \midrule
    6 &    {  reddit}  & \begin{tabular}{@{}l@{}}   $n$: 232,965 \\   $m$: 114,615,892 \end{tabular} & 93.45\% & 94.85\% & 96.07\% \\ \midrule
    7 &    {  enwiki-2024}  & \begin{tabular}{@{}l@{}}   $n$: 6,790,971 \\   $m$: 172,762,484 \end{tabular} & 91.20\% & 91.52\% & 91.78\% \\ \midrule
     8 &   {  eu-2015-host}  & \begin{tabular}{@{}l@{}}   $n$: 11,264,052 \\   $m$: 386,915,963 \end{tabular} & 39.50\% & 39.50\% & 39.68\% \\ \midrule
     9 &   {  arabic-2005}  & \begin{tabular}{@{}l@{}}   $n$: 22,744,080 \\   $m$: 639,999,458 \end{tabular} & 29.88\% & 29.88\% & 30.07\% \\ \midrule
     10 &   {  it-2004}  & \begin{tabular}{@{}l@{}}   $n$: 41,291,594 \\   $m$: 1,150,725,436 \end{tabular} & 30.00\% & 30.00\% & 30.44\% \\ \midrule
     11 &   {  twitter-2010}  & \begin{tabular}{@{}l@{}}   $n$: 41,652,230 \\   $m$: 1,468,365,182 \end{tabular} & 86.92\% & 87.79\% & 87.85\% \\ \midrule
    \end{tabular}
    \caption{Batching preserves reduction power: gap to the stable coloring (corresponding to the batch-free case $N_b =1$) is below 5\%.}
    \label{tab:tablereduction}
\end{table*}

\begin{table*}[t!]
\centering
\begin{tabular}{lrr}
\toprule
\emph{Dataset} & \multicolumn{1}{c}{\emph{True}} & \multicolumn{1}{c}{\emph{Computed}} \\
\midrule
ljournal-2008 & 4,295,622 & 4,295,308 \\
eu-2015-host & 4,449,148 & 4,443,240 \\
arabic-2005 & 6,795,925 & 6,639,368 \\
it-2004 & 12,385,929 & 11,582,482 \\
twitter-2010 & 36,204,326 & 35,921,554 \\
\bottomrule
\end{tabular}
\caption{As the scale of the graph increase the partition computed by the power-iteration method of ~\cite{10.5555/2892753.2892817} starts to diverge from the True coarsest stable coloring due to floating-point precision loss.}
\label{tab:tableinstability}
\end{table*}

\subsubsection{Randomized refinement recovers the coarsest stable
coloring.}
With $N_b=1$, Algorithm~\ref{alg:rbsc} reduces to
Algorithm~\ref{alg:rsc} applied to the full graph.
Theorem~\ref{thm:rsc} guarantees that the result is the coarsest
stable coloring with probability at least $1-10^{-9}$.  We confirm
this empirically: on all eleven benchmark graphs, the partition
returned by RBSC at $N_b=1$ matches, block for block, the
coarsest stable coloring computed by the deterministic baseline~PR (Table~\ref{tab:tablereduction}).
This validates both the probabilistic guarantee and the
integer-arithmetic implementation.

\subsubsection{Numerical instability of floating-point approaches.}
Table~\ref{tab:tableinstability} compares the sizes of the coarsest stable coloring
(computed by PR) with the partition computed by the power-iteration
method of~\cite{10.5555/2892753.2892817}, which operates in
floating-point arithmetic.  On smaller graphs, the two coincide; as
scale increases, however, floating-point precision loss causes the
power-iteration method to \emph{under}-refine, merging nodes that
should be distinguished.  On \texttt{it-2004}, the discrepancy exceeds
$8\times 10^5$ blocks. This numerical instability is the reason we do not include the power-iteration method of~\cite{10.5555/2892753.2892817} as a runtime baseline.

\subsubsection{Batching preserves partition quality.}
Table~\ref{tab:tablereduction} reports reduction power for RBSC at three
batch sizes.  By Theorem~\ref{thm:batch-correct}, every partition
returned by RBSC is a valid stable coloring of the original graph;
however, when $N_b>1$, boundary nodes are conservatively isolated as
singletons (Section~\ref{sec:batching}), so the result may be
strictly finer than the coarsest stable coloring.  The key empirical
question is whether this conservatism matters in practice.

Across all eleven datasets, the partitions obtained at $N_b=2$
and $N_b=4$ are within $5\%$ of the coarsest stable coloring
(the $N_b=1$ column).  On several graphs
(\texttt{cnr-2000}, \texttt{hollywood-2009}, \texttt{arabic-2005},
\texttt{it-2004}), the gap is zero or negligible, indicating that the
iterated quotient contraction fully recovers the coarsest partition
even with multiple batches.  The largest gaps occur on dense social
graphs (\texttt{yelp}, \texttt{reddit}), where many nodes likely have edges
spanning multiple batches. 

\subsection{Runtime Efficiency}
\label{sec:runtime}

\begin{table*}[t]
\centering
\begin{tabular}{@{}rl rr rr rrr@{}}
\toprule
& & \multicolumn{2}{c}{\textbf{Single core}} &
  \multicolumn{2}{c}{\textbf{Multicore}} &
  \multicolumn{3}{c}{\textbf{GPU (RBSC)}} \\
\cmidrule(lr){3-4}\cmidrule(lr){5-6}\cmidrule(l){7-9}
& \textbf{Dataset}
  & PR & RSC-CPU & SR$_{32}$ & SR$_{64}$
  & $N_b{=}1$ & $N_b{=}2$ & $N_b{=}4$ \\
\midrule
1  & {flickr}           & 0.15  & 0.05 & 0.60   & 0.66  & \best{0.01} & 0.06  & 0.09  \\
2  & {cnr-2000}         & 0.27  & 5.01  & 8.92   & 15.09  & \best{0.10} & 0.16  & 0.24  \\
3  & {yelp}             & 4.60  & 0.97  & 19.56  & 13.30  & \best{0.05} & 0.75 & 1.23  \\
4  & {ljournal-2008}    & 34.56 & 40.10 & 181.71 & 175.15 & \best{0.39} & 5.14 & 6.64 \\
5  & {hollywood-2009}   & 21.57 & 2.43 & 109.61 & 67.63  & \best{0.30} & 1.83 & 2.10  \\
6  & {reddit}           & 27.16 & 2.00  & 86.86  & 55.28  & \best{0.25} & 3.67 & 4.49  \\
7  & {enwiki-2024}      & 91.74 & 57.17 & T/O    & T/O    & \best{0.66} & 7.22 & 10.98 \\
8  & {eu-2015-host}     & 62.80 & 189.81 & T/O    & T/O    & \best{1.76} & 3.60 & 10.56 \\
9  & {arabic-2005}      & 72.01 & 3276.20 & T/O    & T/O    & \best{14.54} & 17.44 & 39.84 \\
10 & {it-2004}          & 143.79& T/O   & T/O    & T/O    & \best{121.39} & 133.29 & 190.50 \\
11 & {twitter-2010}     & 935.14& 605.15 & T/O    & T/O    & \best{6.73} &  205.60 & 176.01\\
\bottomrule
\end{tabular}
\caption{Runtime comparison (wall-clock seconds, excluding I/O).
  \textbf{PR}: single-threaded classical partition
  refinement~\cite{DBLP:conf/tacas/ValmariF10}.  \textbf{RSC-CPU}: our
  randomized Algorithm~\ref{alg:rsc} on one CPU core.
  \textbf{SR$_{32}$/SR$_{64}$}: multicore symbolic bisimulation
  reduction~\cite{DBLP:journals/entcs/BlomHKP08,DBLP:conf/tacas/DijkP16} on 32/64
  cores.  \textbf{RBSC}: our Algorithm~\ref{alg:rbsc} on GPU(s);
  $N_b=1$ uses 1~GPU, $N_b=2$ and $N_b=4$ run batches in
  parallel across multiple GPUs.  T/O indicates timeout at 3\,600\,s.
  Best result per dataset highlighted in \best{green}.}
\label{tab:runtime}
\end{table*}

Table~\ref{tab:runtime} reports wall-clock execution times. 

\textbf{Randomized refinement vs.\ classical partition refinement
(CPU).}
\label{sec:numerical}
Comparing RSC-CPU (our Algorithm~\ref{alg:rsc} on a single CPU core)
with PR (classical splitter-based refinement, also single-threaded)
reveals that the linear-algebraic formulation is often faster in
practice despite its higher worst-case complexity ($O(nm)$ vs.\
$O(m\log n)$).  On five of eleven datasets, RSC-CPU outperforms PR,
sometimes by an order of magnitude (\texttt{flickr}: $3\times$,
\texttt{reddit}: $13\times$). 

\textbf{GPU acceleration.}
RBSC at $N_b=1$ (single GPU) consistently outperforms both
single-core CPU baselines, achieving speedups of up to ${\sim}138\times$
over PR and up to ${\sim}225\times$ over RSC-CPU
(\texttt{arabic-2005}). These gains reflect the natural fit between
sparse matrix--vector multiplication and GPU hardware: the
computation is data-parallel, and memory
 bound—precisely the regime where modern high memory bandwidth GPUs excel. 

\textbf{Comparison with multicore baselines.}
On the six datasets where the multicore bisimulation
SR$_{64}$ completed within the timeout, RBSC on a single GPU is
$66\times$--$449\times$ faster.  SR$_{32}$ and SR$_{64}$ timed out on
all remaining datasets.  These results indicate that multicore
partition-refinement approaches, which rely on irregular memory access
patterns and fine-grained synchronization, do not bridge the
performance gap for stable-coloring computations; the linear-algebraic
reformulation is essential for exploiting parallel hardware
effectively.

\textbf{Effect of batch size on runtime.}
Comparing the $N_b=1$, $2$, and $4$ columns of
Table~\ref{tab:runtime} reveals a consistent pattern: processing the
full graph in a single batch is fastest whenever it fits in VRAM.
Reducing the batch size incurs two overheads: (i)~the iterated
quotient contraction requires multiple outer-loop iterations, and
(ii)~boundary-node singletons increase the partition size, slowing
convergence.  Multi-GPU scheduling partially compensates—batches
execute concurrently—but does not fully offset these costs.  The
practical recommendation is therefore straightforward: \emph{use the
largest batch size that fits in device memory}.  Smaller batches
remain viable when VRAM is constrained, though runtime penalties vary widely across datasets ($1.6\times$--$26\times$ at $N_b=4$ vs.\ $N_b=1$), depending on the graph's boundary-node fraction and contraction behavior.

\subsection{Scalability to Massive Graphs}
\label{sec:massive}
\begin{table}[t]
\centering
\small
\begin{tabular}{@{}l l cc cc@{}}
\toprule
& & \multicolumn{2}{c}{$M_b=1.8$\,B} &
    \multicolumn{2}{c}{$M_b=4.3$\,B} \\
\cmidrule(lr){3-4}\cmidrule(l){5-6}
\textbf{Dataset} & \textbf{Size}
  & Red. & Time\,(s)
  & Red. & Time\,(s) \\
\midrule
gsh-2015
  & \scriptsize{989M\,/\,34B}
  & 39.0\% & 837.90
  & 38.8\% & 492.94 \\
{clueweb12}
  & \scriptsize{978M\,/\,43B}
  & 46.4\% & 2\,464
  & 46.1\% & 2\,174 \\
{uk-2014}
  & \scriptsize{788M\,/\,48B}
  & 42.8\% & 5\,195
  & 42.7\% & 3\,143 \\
\bottomrule
\end{tabular}
\caption{Stable coloring on massive web-scale graphs (8~GPUs).  Two
  per-device batch caps are shown: $M_b=1.8$\,B and $M_b=4.3$\,B
  edges.  All CPU baselines (PR, RSC-CPU, SR) timed out or ran out of
  memory.  RBSC is the only method that completes within resource
  limits.  Larger batch caps improve both reduction power and runtime.}
\label{tab:massive}

\end{table}

The primary motivation for batching is not multi-GPU parallelism but
\emph{memory scalability}: enabling stable-coloring computation on
graphs whose edge sets do not fit in the memory of any single device.
Table~\ref{tab:massive} reports results on three web-scale graphs with
$33$--$48$ billion edges, using per-device batch caps of $M_b=1.8$\,B
and $M_b=4.3$\,B edges and concurrent scheduling across all 8~GPUs.

\textbf{Feasibility.}
All CPU baselines (PR, RSC-CPU, SR) either timed out or exhausted
available memory on these instances. So, we cannot directly measure the gap to the optimum. RBSC is the only method that
completes within resource limits, demonstrating that the combination
of randomized refinement and batching enables stable-coloring analysis
at scales that were previously infeasible.  For these experiments, we
employ a more aggressive termination condition: the outer loop of
Algorithm~\ref{alg:rbsc} stops as soon as the number of batches in
the next iteration would not decrease, avoiding diminishing-returns
iterations.

\textbf{Effect of per-device batch cap.}
Increasing the cap from $1.8$\,B to $4.3$\,B edges (the maximum
that fits in VRAM on our H200 GPUs) improves both quality and runtime.
Reduction power improves by $0.1$--$0.3$ percentage points, and
wall-clock time decreases by $1.14$--$1.69\times$.  The improvement
is explained by the same mechanism identified on medium-scale
graphs: larger batches reduce the fraction of boundary nodes, leading
to more effective aggregation in each round and faster convergence of
the quotient contraction. 

\section{Conclusion}
We introduced a GPU-amenable reformulation of 1-WL stable coloring, combining randomized linear-algebraic refinement over $\mathbb{Z}/2^{64}\mathbb{Z}$ with correctness-preserving batching for graphs beyond device memory. The randomized algorithm gives probabilistic guarantees and avoids the numerical instability of floating point methods at scale. The batching scheme decomposes the edge set into independently processable subgraphs whose refinements provably yield a global stable coloring; on all benchmarks, batched partitions stay within 5\% of the coarsest stable coloring. A single-GPU implementation achieves speedups of up to ${\sim}138\times$ over classical partition refinement and ${\sim}449\times$ over multicore baselines. On web-scale graphs with over 30 billion edges, where CPU baselines fail, our method computes stable colorings with near-complete reduction, enabling 1-WL expressiveness analysis at this scale. 

The correspondence between 1-WL stable coloring and equitable partitions suggests future work into leveraging our method for downstream graph ML tasks, e.g. \cite{DBLP:conf/icdm/SquillaceTTV24}; the relationship with linear invariance in dynamical systems, instead, offers opportunities for lifting randomized GPU-based techniques to massive-scale model-reductions for nonlinear systems~\cite{cardelli2017erode,doi:10.1073/pnas.1702697114}.

%
%
%
\bibliographystyle{splncs04}
\bibliography{ref}

%
%
%
%
%
\newpage

\onecolumn

\title{Scaling Weisfeiler–Leman Expressiveness Analysis to Massive Graphs with GPUs \\ (Appendix)}
\author{Filippo Biondi\inst{1} \corr \and Mirco Tribastone\inst{1} \and Max Tschaikowski\inst{2}}
\institute{
IMT Lucca, Italy\\
\email{\{filippo.biondi,mirco.tribastone\}@imtlucca.it}
\and
Sapienza University Rome, Italy\\
\email{max.tschaikowski@uniroma1.it}
}

\maketitle

\appendix

\section{Proofs}
\label{app:proofs}

\begin{proof}[Theorem~\ref{thm:rsc}]
Define the degree one polynomials $p_i$ in $k = |H|$ variables by setting
\begin{align*}
p_i(w^H) & = \sum_{B \in H} a_{i,B} w_B^H, \qquad \text{with} \quad 
a_{i,B} = \sum_{l \in B} a_{i,l} , \quad B \in H, \ w^H \in \RE^H .
\end{align*}
Let us assume that there exists a $w \in \RE^n$ that is symmetric on $H$ and for which $(Aw)_i \neq (Aw)_j$. This implies that there is a $w^H \in \RE^k$ such that $p_i(w^H) \neq p_j(w^H)$, i.e., the degree one polynomial $p_{i,j} = p_i - p_j$ in $k$ variables is not identically zero over $\RE^k$. By assumption on the row sums of $A$, it holds that $0 \leq a_{l,B} \leq 2^{q/2}$ for all $B \in H$ and $1 \leq l \leq n$. This, in turn, implies that $p_{i,j}$ is not identically zero over $R = \ZZ / 2^{q} \ZZ$ either. This is because for any pair $a_{i,B} \neq a_{j,B}$ (over $\RE$), it holds that $\gcd(a_{i,B} - a_{j,B},2^{q}) = 2^t$ with $t \leq q/2$. By Lemma~\ref{lem:group}, this also implies that the number of roots of $p_{i,j}$ in $R^k$ is at most $2^{q/2} (2^q)^{k-1}$. Consequently, picking a random $w^H \in R^k$, the probability of hitting a root of $p_{i,j}$ in $R^k$ is at most
\[
2^{q/2} (2^q)^{k-1} / (2^q)^k = 2^{q/2} / 2^{q} = 2^{-q/2} \leq 10^{-9} 
\]
if $q \geq 64$. Recall that we wish to estimate the probability of hitting a root of $p_{i,j}$ in $R^k$ by picking randomly a $w^H \in P(k)$, where $P(k)$ denotes all $w^H \in R^k$ with pairwise different coordinates. By Lemma~\ref{eq:lemma:bound}, we infer that $|P(k)| \geq \tfrac{6}{10} r^k$. Consequently, picking a random $w^H \in P(k)$, the probability of hitting a root of $p_{i,j}$ is at most 
\[
2^{q/2} (2^q)^{k-1} / |P(k)| \leq 2^{q/2} / (\tfrac{6}{10} \cdot 2^{q}) \leq 10^{-9} 
\]
for $q \geq 64$. Overall, we can bound the error probability that a randomly picked $w^H \in P(k)$ happens to be a root of $p_{i,j}$, even though $p_{i,j}$ itself is not zero (hence, coordinates $i,j$ should be refined). Thus, if $|H|$ is not a BE, there will be at least one such pair $i,j$, and the error probability of missing a split is at most $\varepsilon = 10^{-9}$. With this crucial insight, we can apply the law of total probability, bounding the overall error probability of the algorithm:
\[
\Pr(\text{error})
= \sum_{t=1}^{n} \Pr(\text{error} \mid T = t)\Pr(T = t) \leq \sum_{t=1}^{n} \varepsilon  \Pr(T = t) = \varepsilon \cdot 1 ,
\]
where $T \in \{1,\ldots,n\}$ is the random variable describing the number of while loop iterations.

In the last part, we estimate the worst case complexity of the algorithm. 
Since each refinement step is obtained by means of a matrix-vector product, a single refinement can be done in $\calO(m)$ (to simplify presentation, we assume without loss of generality that that $n \log(n) \leq m$). Instead, the number of steps needed to sample a random vector $w \in R^n$ is $\calO(n)$, while checking whether it has $|H|$ different values can be done in $\calO(n \log(n))$ by means of sorting. Since the expected number of such samplings can be bounded by $2 = (\tfrac{1}{2})^{-1}$ via Lemma~\ref{eq:lemma:bound}, the expected worst case complexity of one while loop iteration is thus $\calO(m)$. Since there can be at most $n$ refinement steps, the entire algorithm needs at most $\calO(nm)$ steps. For the space complexity, instead, we note that partitions can be compactly stored as vectors of length $n$ and the coarsest partition on which two given vectors are symmetric can be computed via lexicographic sorting in $\calO(n \log(n))$. This completes the proof.  
\end{proof}

\begin{lemma}\label{lem:group}
Let $R := \mathbb{Z}/ r \mathbb{Z}$ with $r = 2^q$ and let $k \ge 1$. Consider a linear polynomial $p(x_1,\dots,x_k) = a_1 x_1 + \ldots + a_k x_k$ over $R$. Fix $t \geq 1$ so that it holds $\gcd(a_1,\ldots,a_k,2^q) = g = 2^t$. Then, $p$ has $2^t r^{k-1}$ roots in $R$.
\end{lemma}

\begin{proof}
Define the group homomorphism $\varphi : (\mathbb Z/r\mathbb Z)^k \to \mathbb Z/r\mathbb Z$ with $\varphi(x)=\sum_{i=1}^k a_i x_i$. Then the roots of $p$ are exactly $\ker(\varphi)$, hence $|\{x\in(\mathbb Z/r\mathbb Z)^k \mid p(x)=0\}| = |\ker(\varphi)|$. We now claim $\operatorname{im}(\varphi)=g\cdot(\mathbb Z/r\mathbb Z)$, that is, the image is the set of multiples of $g$ modulo $r$. To see this, we first note that  $g\mid a_i$, hence every value $\sum_{i=1}^k a_i x_i$ is divisible by $g$, implying that $\operatorname{im}(\varphi)\subseteq g(\mathbb Z/r\mathbb Z)$. Conversely, Bézout's identity gives integers $u_1,\dots,u_k,v$ such that $u_1 a_1 + \cdots + u_k a_k + vr = g$. Reducing modulo $r$ yields $\sum_{i=1}^k u_i a_i \equiv g \pmod r$, so $g\in \operatorname{im}(\varphi)$. By scaling the inputs $u_i$ by $l$, we obtain $l g\in \operatorname{im}(\varphi)$ for all $l$, hence $g(\mathbb Z/r \mathbb Z)\subseteq \operatorname{im}(\varphi)$, showing overall that $g(\mathbb Z/r \mathbb Z) = \operatorname{im}(\varphi)$. Consider next $\mu_g:\mathbb Z/r\mathbb Z \to \mathbb Z/r\mathbb Z$ given by $\mu_g(x)=g x$. Then $\operatorname{im}(\mu_g)=g(\mathbb Z/r\mathbb Z)$ and $\ker(\mu_g) = \{0, s, 2s, \ldots, (g-1)s\}$ for $s = r/g$, which is an integer value because $g | r$ by assumption. The first isomorphism theorem then gives 
\[
| (\mathbb Z/r\mathbb Z) / \ker(\mu_g) | = | \operatorname{im}(\mu_g)| , \ \text { yielding } \ 
|g(\mathbb Z/r\mathbb Z)|
=
|\operatorname{im}(\mu_g)|
=
\frac{|\mathbb Z/r\mathbb Z|}{|\ker(\mu_g)|}
=
\frac{r}{g}.
\]
In turn, since $|(\mathbb Z/r\mathbb Z)^k|=r^k$, another application of the first isomorphism theorem gives $| (\mathbb Z/r\mathbb Z)^k / \ker(\varphi) | = |\operatorname{im}(\varphi)|$, yielding
\[
|\ker(\varphi)|
=
\frac{|(\mathbb Z/r\mathbb Z)^k|}{|\operatorname{im}(\varphi)|}
=
\frac{|(\mathbb Z/r\mathbb Z)^k|}{|g(\mathbb Z/r\mathbb Z)|}
=
\frac{r^k}{r/g}
=
g\,r^{k-1}.
\]
\end{proof}

\begin{lemma}\label{eq:lemma:bound}[Birthday Paradox]
Let $1 \leq k \leq \sqrt{r}$ be an integer with $r = 2^q$ and set $\wp(k) = r(r-1)\cdots (r-k+1) = \prod_{i=0}^{k-1} (r-i)$ denote the cardinality of $P(k)$ in the proof above. Then, for $q \geq 64$, it holds that 
\[
\wp(k) / r^k \geq 0.60
\]
\end{lemma}

\begin{proof}
We write
\[
\frac{\wp(k)}{r^k}
=
\prod_{i=0}^{k-1}\left(1 - \frac{i}{r}\right).
\]
For $0 \le x < 1$ we use the inequality $\log(1-x) \ge -\frac{x}{1-x}$. Since $0 \le i \le k-1 \le  \sqrt{r}$, we have 
\[
\frac{1}{1 - \frac{i}{r}}
\leq
\frac{1}{1 - 2^{-q/2}} =: c 
\quad \text{ and } \quad 
\log\!\left(1 - \frac{i}{r}\right)
\ge
-\frac{\frac{i}{r}}{1 - \frac{i}{r}}
\ge
-\frac{ci}{r}.
\]

Summing over $i=0,\ldots,k-1$, we obtain
\begin{multline*}
\log \left( \frac{\wp(k)}{r^k} \right)
=
\sum_{i=0}^{k-1}
\log\!\left(1 - \frac{i}{r}\right) 
\ge
\sum_{i=0}^{k-1}
\left(-\frac{ci}{r}\right) \\
=
-\frac{c}{r} \sum_{i=0}^{k-1} i 
= 
-\frac{c}{r} \cdot \frac{k(k-1)}{2} 
\geq - \frac{1}{2} - \calO(2^{q/2}) 
\end{multline*}
Taking exponents on both sides, we obtain as claimed $\wp(k) / r^k \geq 0.60$.
\end{proof}

\begin{proof}[Theorem~\ref{thm:batch-correct} and~\ref{thm:rbsc-complexity}]
We first observe that an $A \in \mathbb{N}_0^{n \times n}$ whose row sums do not exceed $2^{q/2}$ implies that any aggregated matrix $A'$ arising from $A$ will enjoy the same property. This is because an aggregated matrix $A'$ with respect to a partition $H$ is obtained by dropping rows in $A$ (no effect on row sums) and by summing the columns of matrix $A$ according to the blocks of $H$ (preservation of row sums). With this, let us assume that $H$ has been computed by Algorithm~\ref{alg:rbsc}. To show that $H$ is indeed a BE, pick any $B,B' \in H$ and $i,j \in B$. We need to show that 
\begin{align}\label{eq:proof:external}
\sum_{k \in B'} a_{i,k} = \sum_{k \in B'} a_{j,k}
\end{align}
In the case $B$ contains a boundary node, $B$ must be a singleton block 
and there is nothing to prove because $i = j$. Let us thus assume that $B$ has no boundary nodes. In this case, block $B$ was add to $H$. 
Hence, there is a unique $\kappa$ such that $i,j \in V_\kappa$ and
\begin{align*}
\sum_{k \in B' : (i,k) \in E_\kappa} a_{i,k} = \sum_{k \in B' : (j,k) \in E_\kappa} a_{j,k} 
\end{align*}
Here, we implicitly assumed that the result of the randomized Algorithm~\ref{alg:rsc} was indeed correct (the formal error estimation follows by means of the law of total probability as in the proof of Theorem~\ref{thm:rsc}). Moreover, due to the construction of $V_\kappa$, for any $S \subseteq V$, it holds that
\[
\sum_{k \in S} a_{i,k} = \sum_{k \in S : (i,k) \in E_\kappa} a_{i,k} .
\]
A combination of the last two statements yields~(\ref{eq:proof:external}). 

For the proof of complexity, we first estimate the complexity of one until loop iteration. To this end, we note that batches $E_1, \ldots, E_{N_b}$ can be computed in $\calO(m)$ time and space. Likewise, $V_1,\ldots,V_{N_b}$ can be computed in $\calO(m)$ and $\calO(n+m)$ space by adding for each edge $(i,j) \in E$ its batch label $k$ (when $(i,j) \in E_k$) to the node it originates from, i.e., we add $k$ to a list associated to node $i$; after iterating through all edges, the inner nodes are these which have exactly one batch label. For the complexity of the for loop, we first note that the for loop comprises $N_b$ invocations of RSC for matrices with at most $M_b$ non-zero entries. This gives, respectively, the expected time and space complexity of $\calO(nm)$ and $\calO(n+m)$ by Theorem~\ref{thm:rsc} because $\sum_k |E_k| = m$. Additionally to that, the updates of $H_k^0$ and $H$ yield time complexity $\calO(n)$, which is not dependent on $N_b$. This can be realized by storing partitions as vectors and by updating only the part of the vector that is required in the current for loop iteration. For instance, each $H_k^0$ can be constructed by replacing $|V_k|$ entries of the vector $(1,2,3,\ldots,n)^T$ which encodes partition $\{ \{i\} \mid 1 \leq i \leq n\}$. A similar remark applies to updates of partition $H$. Since the aggregated matrix can be computed in $\calO(nm)$ steps (for each block  representative $i_B$, find all its outgoing edges in $A$ and sum them according to the blocks of partition $H$), the time and space complexity of one until loop iteration can be overall bounded by $\calO(nm)$ and $\calO(m)$, respectively. Since there are at most $n$ iterations of the until loop (this occurs if every batching reduces the graph by exactly one node), we obtain the complexity statement.
\end{proof}

\section{Implementation}

Both Algorithms~\ref{alg:rbsc} and~\ref{alg:rsc} are implemented in C++ with the CUDA toolkit \cite{10.1145/1401132.1401152}. Algorithm~\ref{alg:rbsc} is a CPU algorithm while the inner iteration of Algorithm~\ref{alg:rsc} has GPU sections; batch iterations in Algorithm~\ref{alg:rbsc} are executed in parallel using the OpenMP library to exploit all the GPUs available in the system \cite{10.1109/99.660313}. At the start of Algorithm~\ref{alg:rsc}, initial partition $H_k^0$ is transferred as a vector $y$ to the GPU memory, where nodes $i,j$ are in the same block of $H_k^0$ if and only if $y_i = y_j$; instead, the matrix-vector product of Algorithm~\ref{alg:rsc} is carried out using the device-wide reduction primitives of the CUB library. A custom CUDA kernel is used to compute the  
random entries of $w$ according to  the entries of $y$. Specifically, random entries are obtained by applying the \texttt{xoshiro256**} random number generator to the output of a \texttt{SplitMix64} generator with $y$ as the vector of seeds, i.e.
\[
w_i = \text{\texttt{xoshiro256**}(\texttt{SplitMix64}($y_i$))}
\]
for all $i$, where each $w_i$ is stored as an unsigned 64 bit integer (thus yielding $r = 2^{64}$ in~Algorithm~\ref{alg:rsc}), 
see~\cite{10.1145/3460772,10.1145/2714064.2660195} for more information on the routines.

\end{document}